\documentstyle[12pt]{article}

\font\twlgot =eufm10 scaled \magstep1 \font\egtgot =eufm8
\font\sevgot =eufm7 \font\twlmsb =msbm10 scaled \magstep1
\font\egtmsb =msbm8 \font\sevmsb =msbm7

\newfam\gotfam
\def\pgot{\fam\gotfam\twlgot}
\textfont\gotfam\twlgot \scriptfont\gotfam\egtgot
\scriptscriptfont\gotfam\sevgot
\def\got{\protect\pgot}
\newfam\msbfam
\textfont\msbfam\twlmsb \scriptfont\msbfam\egtmsb
\scriptscriptfont\msbfam\sevmsb
\def\Bbb{\protect\pBbb}
\def\pBbb{\relax\ifmmode\expandafter\Bb\else\typeout{You cann't use
Bbb in text mode}\fi}
\def\Bb #1{{\fam\msbfam\relax#1}}

\newcommand{\gS}{{\got S}}
\newcommand{\gF}{{\got F}}
\newcommand{\ccG}{{\got g}}

\def\thebibliography#1{\bigskip\section*{}\bigskip\list
{$^{\arabic{enumi}}$}{\settowidth\labelwidth{#1}\leftmargin\labelwidth
\advance\leftmargin\labelsep
\usecounter{enumi}}
\def\newblock{\hskip .11em plus .33em minus .07em}
\sloppy\clubpenalty4000\widowpenalty4000 \sfcode`\.=1000\relax}

\def\op#1{\mathop{\fam0 #1}\limits}

\newcommand{\Ker}{{\rm Ker\,}}
\newcommand{\im}{{\rm Im\,}}

\newcommand{\beq}{\begin{equation}}
\newcommand{\eeq}{\end{equation}}
\newcommand{\ben}{\begin{eqnarray}}
\newcommand{\een}{\end{eqnarray}}
\newcommand{\be}{\begin{eqnarray*}}
\newcommand{\ee}{\end{eqnarray*}}
\newcommand{\bea}{\begin{eqalph}}
\newcommand{\eea}{\end{eqalph}}
\newcommand{\cA}{{\cal A}}

\newcommand{\cP}{{\cal P}}

\newcommand{\cL}{{\cal L}}

\newcommand{\cE}{{\cal E}}

\newcommand{\cF}{{\cal F}}
\newcommand{\cS}{{\cal S}}
\newcommand{\cC}{{\cal C}}
\newcommand{\cO}{{\cal O}}

\newcommand{\cG}{{\cal G}}
\newcommand{\bL}{{\bf L}}

\newcommand{\bu}{{\bf u}}
\newcommand{\bb}{{\bf b}}

\newcommand{\al}{\alpha}

\newcommand{\bt}{\beta}
\newcommand{\dl}{\delta}
\newcommand{\la}{\lambda}
\newcommand{\La}{\Lambda}
\newcommand{\f}{\phi}

\newcommand{\m}{\mu}

\newcommand{\g}{\gamma}

\newcommand{\th}{\theta}

\newcommand{\vt}{\vartheta}
\newcommand{\vf}{\varphi}
\newcommand{\up}{\upsilon}

\newcommand{\di}{{\rm dim\,}}

\newcommand{\si}{\sigma}
\newcommand{\Si}{\Sigma}
\newcommand{\w}{\wedge}
\newcommand{\wt}{\widetilde}

\newcommand{\ol}{\overline}
\newcommand{\dr}{\partial}
\newcommand{\ar}{\op\longrightarrow}
\newcommand{\llr}{\op\longleftarrow}
\newcommand{\lto}{\leftarrow}
\newcommand{\ot}{\otimes}

\newcommand{\ve}{\varepsilon}

\newcommand{\rdr}{\stackrel{\leftarrow}{\dr}{}}

\let\ssection=\section
\renewcommand{\section}{\setcounter{equation}{0}\ssection}

\newcounter{eqalph}
\newcounter{equationa}
\newcounter{remark}
\newcounter{example}
\newcounter{theorem}
\newcounter{proposition}
\newcounter{lemma}
\newcounter{corollary}
\newcounter{definition}
\newenvironment{eqalph}{\stepcounter{equation}
\setcounter{equationa}{\value{equation}} \setcounter{equation}{0}

\begin{eqnarray}}{\end{eqnarray}\setcounter{equation}{\value{equationa}}}

\def\theremark{\arabic{remark}}

\def\thetheorem{\arabic{theorem}}

\newenvironment{proof}{
{\it Proof:}}{}

\newenvironment{theo}{\refstepcounter{theorem}
{\bf Theorem \thetheorem:}}{}
\newenvironment{prop}{\refstepcounter{theorem}
{\bf Proposition \thetheorem:}}{}

\newenvironment{defi}{\refstepcounter{theorem}
{\bf Definition \thetheorem:}}{}
\newenvironment{cond}{\refstepcounter{theorem}
{\bf Condition \thetheorem:}}{}

\newcommand{\mar}[1]{}

\begin{document}
\hbox{}

{\parindent=0pt

{\large\bf On the notion of gauge symmetries of generic Lagrangian
field theory}
\bigskip

{\sc G.Giachetta}\footnote{Electronic mail:
giovanni.giachetta@unicam.it}

{\sl Department of Mathematics and Informatics, University of
Camerino, 62032 Camerino (MC), Italy}

\medskip

{\sc L.Mangiarotti}\footnote{Electronic mail:
luigi.mangiarotti@unicam.it}

{\sl Department of Mathematics and Informatics, University of
Camerino, 62032 Camerino (MC), Italy}

\medskip

{\sc G. Sardanashvily}\footnote{Electronic mail:
sardanashvi@phys.msu.ru}

{\sl Department of Theoretical Physics, Moscow State University,
117234 Moscow, Russia}

\bigskip
\bigskip

General Lagrangian theory of even and odd fields on an arbitrary
smooth manifold is considered. Its non-trivial reducible gauge
symmetries and their algebra are defined in this very general
setting by means of the inverse second Noether theorem. In
contrast with gauge symmetries, non-trivial Noether and
higher-stage Noether identities of Lagrangian theory can be
intrinsically defined by constructing the exact Koszul--Tate
complex. The inverse second Noether theorem that we prove
associates to this complex the cochain sequence with the ascent
operator whose components define non-trivial gauge and
higher-stage gauge symmetries. These gauge symmetries are said to
be algebraically closed if the ascent operator can be extended to
a nilpotent operator. The necessary conditions for this extension
are stated. The characteristic examples of Yang--Mills supergauge
theory, topological Chern--Simons theory, gauge gravitation theory
and topological BF theory are presented.

}

\bigskip
\bigskip

\noindent {\bf I. INTRODUCTION}
\bigskip

Treating gauge symmetries of Lagrangian field theory, one is
traditionally based on an example of the Yang--Mills gauge theory
of principal connections on a principal bundle $P\to X$ with a
structure Lie group $G$. In this theory, gauge transformations are
defined as vertical automorphisms of $P$. Infinitesimal generators
of one-parameter groups of these automorphisms are $G$-invariant
vertical vector fields on $P$. They are identified with global
sections $\xi$ of the quotient $VP/G$ of the vertical tangent
bundle $VP$ of $P\to X$ with respect to the right action of $G$ on
$P$. These sections take a local form $\xi=\xi^p(x)e_p$ where
$\{e_p\}$ is the basis for the Lie algebra $\ccG$ of $G$. They
constitute a Lie $C^\infty(X)$-algebra $\ccG(X)$ with respect to
the bracket
\mar{804}\beq
[\xi,\eta]=c^r_{pq}\xi^p\eta^q e_r, \label{804}
\eeq
where $c^r_{pq}$ are structure constants of $\ccG$. Being
$G$-equivariant, principal connections on a principal bundle $P$
are represented by global sections of the quotient
\mar{s0}\beq
C=J^1P/G \label{s0}
\eeq
of the jet bundle $J^1P$ of $P$ which is coordinated by $(x^\m,
a^r_\m)$.$^1$ Vertical automorphisms of $P$ yield automorphisms of
the bundle $C$ (\ref{s0}). Infinitesimal generators of
one-parameter groups of these automorphisms are given by vector
fields
\mar{801}\beq
u_\xi=(\dr_\la\xi^r + c^r_{pq}a^p_\la\xi^q)\frac{\dr}{\dr a^r_\la}
\label{801}
\eeq
on $C$. They form a real Lie algebra
\mar{805}\ben
&& u_\xi + u_\eta= u_{\xi+\eta}, \qquad \la u_\xi=u_{\la\xi}, \qquad
\la\in\Bbb R, \nonumber\\
&& [u_\xi,u_\eta]=u_{[\xi,\eta]}, \label{805}
\een
which is isomorphic to the Lie algebra (\ref{804}) seen as a real
algebra (but not the $C^\infty(X)$-one because $u_{f\xi}\neq
fu_\xi$, $f\in C^\infty(X)$). This isomorphism is a linear
differential operator on sections of $VP/G\to X$. The vector
fields (\ref{801}) are exact symmetries of the Yang--Mills
Lagrangian $L_{YM}$, i.e., the Lie derivative of $L_{YM}$ along
any $u_\xi$ (\ref{801}) vanishes. They are called the gauge
symmetries of $L_{YM}$ depending on gauge parameters $\xi\in
\ccG(X)$.

This notion of gauge symmetries has been generalized to Lagrangian
field theory on any fiber bundle $Y\to X$ over an $n$-dimensional
smooth manifold $X$ as follows.$^{3,4}$ Given a $k$-order jet
manifold $J^kY$ of $Y$, let us consider the pull-back
\be
T^kY=TY\op\times_Y J^kY
\ee
of the tangent bundle $TY$ of $Y$ onto $J^kY$ over $Y$. Sections
of $T^kY\to J^kY$ are called generalized vector fields on $Y$.$^2$
A generalized vector field $u$ is said to be a variational
symmetry of a Lagrangian $L$ if the Lie derivative of $L$ along
$u$ is a variationally trivial Lagrangian. Variational symmetries
constitute a real subspace $\cG_L$ of the $C^\infty(J^kY)$-module
of generalized vector fields.

\begin{defi} \mar{s7} \label{s7}
Let $E\to X$ be a vector bundle and $E(X)$ the $C^\infty(X)$
module $E(X)$ of sections of $E\to X$. Let $\zeta$ be a linear
differential operator on $E(X)$ taking values in the space $\cG_L$
of variational symmetries. Elements $u_\f=\zeta(\f)$ of $\im\zeta$
are called gauge symmetries of a Lagrangian $L$ parameterized by
sections $\f$ of $E\to X$.
\end{defi}

Equivalently, these gauge symmetries are given by a section
$\wt\zeta$ of a fiber bundle
\be
TY\op\times_Y J^kY\op\times_Y (J^mE\op\times_X Y)\to
J^kY\op\times_Y (J^mE\op\times_X Y)
\ee
such that $u_\f=\zeta(\f)=\wt\zeta\circ\f$ for any section $\f$ of
$E\to X$.

Gauge symmetries possess the following two important properties.

(i) Let $E'\to X$ be another vector bundle and $\zeta'$ a linear
$E(X)$-valued differential operator on the $C^\infty(X)$-module
$E'(X)$ of sections of $E'\to X$. Then $u_{\zeta'(\vf)}
=(\zeta\circ\zeta')(\vf)$ are also gauge symmetries of $L$ which
factorize through the gauge symmetries $u_\f$.

(ii) The direct and inverse second Noether theorems associate to
gauge symmetries the Noether identities (henceforth NI), which
variational derivatives of $L$ satisfy.$^{3,4}$

Definition \ref{s7} of gauge symmetries can be generalized as
follows.

\begin{defi} \mar{s8} \label{s8}
Let a differential operator $\zeta$ in Definition \ref{s7} need
not be necessarily linear. Then elements of $\im\zeta$ are called
generalized gauge symmetries.
\end{defi}

However, the Noether theorems fail to hold for generalized gauge
symmetries. Definition \ref{s7} of gauge symmetries has been
extended to Lagrangian theory of odd fields by replacement of
$C^\infty(X)$-modules and fiber bundles with Grassmann-graded
$C^\infty(X)$-modules and graded manifolds whose bodies are fiber
bundles, respectively.$^{2,4}$

Gauge symmetries of Lagrangian field theory are thought to
characterize its degeneracy. A problem is that any Lagrangian
possesses gauge symmetries and, therefore, one must separate them
into the trivial and non-trivial ones. Moreover, gauge symmetries
can be reducible, i.e., $\Ker\zeta\neq 0$. Let there exist a
vector bundle $E_1\to X$ and a linear differential operator
$\zeta_1$ on sections of $E_1\to X$ taking values in $\Ker \zeta$.
Elements $\zeta_1(\f_1)$ of $\im\zeta_1$ are called the
first-stage gauge symmetries whose gauge parameters are sections
$\f_1$ of $E_1\to X$.  Since first-stage gauge symmetries in turn
can be reducible, second-stage gauge symmetries are defined, and
so on. Higher-stage gauge symmetries must also be separated into
the trivial and non-trivial ones. This is important because
non-trivial gauge and higher-stage gauge symmetries define the
BRST extension of original Lagrangian field theory for the purpose
of its quantization.$^{5-8}$

Another problem is that gauge symmetries need not form an
algebra.$^{5,9,10}$ The Lie bracket $[u_\f,u_{\f'}]$ of gauge
symmetries $u_\f,u_{\f'}\in \im \zeta$ is a variational symmetry,
but it need not belong to $\im \zeta$.

In contrast with gauge symmetries, non-trivial NI and higher-stage
NI of Lagrangian field theory are well described in homology
terms.$^{8,11}$ Therefore, we define non-trivial gauge and
higher-stage gauge symmetries as those associated to complete
non-trivial NI and higher-stage NI in accordance with the inverse
second Noether theorem (Definitions \ref{s4} and \ref{s6}).

Lagrangian theory of even and odd fields on an $n$-dimensional
smooth real manifold $X$ is adequately formulated in terms of the
Grassmann-graded variational bicomplex.$^{2,4,6,12}$ In accordance
with general theory of NI of differential operators,$^{13}$ NI of
Lagrangian theory are represented by cycles of the chain complex
(\ref{v042}), whose boundaries are treated as trivial NI and whose
homology describes non-trivial NI modulo the trivial
ones.$^{8,11}$ Lagrangian field theory is called degenerate if its
Euler--Lagrange operator satisfies non-trivial NI. The latter obey
first-stage NI, and so on. To describe $(k+1)$-stage NI, let us
assume that non-trivial $k$-stage NI are generated by a projective
$C^\infty(X)$-module $\cC_{(k)}$ of finite rank, whose elements
are called complete NI. In this case, $(k+1)$-stage NI are
represented by $(k+2)$-cycles of the chain complex (\ref{v94})
where $N=k$. If a certain homology condition (Condition
\ref{v155}) holds, trivial $(k+1)$-stage NI are identified with
$(k+2)$-boundaries of this complex. Then its $(k+2)$-homology
describes non-trivial $(k+1)$-stage NI. Degenerate Lagrangian
field theory is called $k$-stage reducible if there exist
non-trivial $k$-stage NI, but all $(k+1)$-stage NI are trivial. In
this case, the chain complex (\ref{v94}) where $N=k$ is exact. It
is called the Koszul--Tate (henceforth KT) complex. The
nilpotentness of its boundary operator (\ref{v92}) is equivalent
to all complete non-trivial NI and higher-stage NI.$^{8,11}$

Recall that the notion of reducible NI has come from that of
reducible constraints,$^{6,14}$ but NI unlike constraints are
differential equations. Therefore, the regularity condition for
the KT complex of constraints is replaced with homology Condition
\ref{v155}.$^{8,11}$

For the sake of simplicity, we here restrict our consideration to
finitely reducible Lagrangian field theory which possesses no
non-trivial $(N+1)$-stage Noether identities for some integer $N$.
In this case, the KT operator (\ref{v92}) and the gauge operator
(\ref{w108'}) contain finite terms.

Different variants of the second Noether theorem have been
suggested in order to relate reducible NI and gauge
symmetries.$^{3,4,6,15}$ Formulated in homology terms, the inverse
second Noether theorem (Theorem \ref{w35}) associates to the KT
complex (\ref{v94}) the cochain sequence (\ref{w108}) with the
ascent operator $\bu$ (\ref{w108'}).$^8$  We define complete
non-trivial gauge and higher-stage gauge symmetries of Lagrangian
field theory as components of this ascent operator, called the
gauge operator (Section IV). The gauge operator unlike the KT one
is not nilpotent, unless non-trivial gauge symmetries are abelian.
This is the cause why an intrinsic definition of non-trivial gauge
and higher-stage gauge symmetries meets difficulties. Defined by
means of the inverse second Noether theorem, non-trivial gauge and
higher-stage gauge symmetries are parameterized by odd and even
ghosts, but not gauge parameters. Herewith, $k$-stage ghosts form
the $(\op\w^nT^*X)$-duals of the modules $\cC_{(k+1)}$, and a
$k$-stage gauge symmetry acts on $(k-1)$-stage ghosts.

For instance, the gauge operator of the gauge symmetries
(\ref{801}) reads
\mar{s10}\beq
\bu= (d_\la c^r +c^r_{pq}a^p_\la c^j )\frac{\dr}{\dr a_\la^r},
\label{s10}
\eeq
where odd ghosts $c^r$ are the generating elements of the exterior
Grassmann algebra $\w\ccG^*$ of the Lie coalgebra $\ccG^*$. This
gauge operator is not nilpotent, unless the Lie algebra $\ccG$ is
commutative. However, $\bu$ (\ref{s10}) is extended to the
nilpotent operator
\be
\bb=\bu +\g=\bu  -\frac12 c^r_{ij}c^ic^j\frac{\dr}{\dr c^r}
\ee
by means of an additional summand $\g$ acting on ghosts. This
nilpotent extension exists because gauge symmetries (\ref{801})
form the Lie algebra (\ref{805}). It is the well known BRST
operator in quantum gauge theory.

Generalizing this example, we say that gauge and higher-stage
gauge symmetries are algebraically closed if the gauge operator
$\bu$ (\ref{w108'}) admits the nilpotent BRST extension $\bb$
(\ref{w109}) where $k$-stage gauge symmetries are extended to
$k$-stage BRST transformations acting both on $(k-1)$-stage and
$k$-stage ghosts (Section V). We show that this nilpotent
extension exists only if the higher-stage gauge symmetry
conditions hold off-shell (Proposition \ref{lmp6}) and only if the
Lie bracket of gauge symmetries is a generalized gauge symmetry
factorizing through these gauge symmetries (Proposition
\ref{830}). For instance, this is the case of abelian reducible
Lagrangian theories and irreducible Lagrangian theories whose
gauge symmetries form a Lie algebra. In abelian reducible
theories, the gauge operator $\bu$ itself is nilpotent.

In Sections VI -- IX, the following characteristic examples are
considered: (i) Yang--Mills supergauge theory exemplifying theory
of odd fields, (ii) topological Chern--Simons theory where some
gauge symmetries become trivial if $\di X=3$, (iii) gauge
gravitation theory whose gauge symmetries are general covariant
transformations, and (iv) topological BF theory with reducible
gauge symmetries.

\bigskip
\bigskip

\noindent {\bf II. GRASSMANN-GRADED LAGRANGIAN FIELD THEORY}
\bigskip

As was mentioned above, Lagrangian theory of even and odd fields
is adequately formulated in terms of the variational bicomplex on
fiber bundles and graded manifolds.$^{2,4,12}$ Let us consider a
composite bundle $F\to Y\to X$ where $F\to Y$ is a vector bundle
provided with bundle coordinates $(x^\la, y^i, q^a)$. Jet
manifolds $J^rF$ of $F\to X$ are also vector bundles $J^rF\to
J^rY$ coordinated by $(x^\la, y^i_\La, q^a_\La)$, $0\leq |\La|\leq
r$, where $\La=(\la_1...\la_k)$, $|\La|=k$, denote symmetric
multi-indices. For the sake of convenience, the value $r=0$
further stands for $F$ and $Y$. Let $(J^rY,\cA_r)$ be a graded
manifold whose body is $J^rY$ and whose structure ring $\cA_r$ of
graded functions consists of sections of the exterior bundle
\be
\w (J^rF)^*=\Bbb R\op\oplus (J^rF)^*\oplus\op\w^2
(J^rF)^*\op\oplus\cdots,
\ee
where $(J^rF)^*$ is the dual of $J^rF\to J^rY$. The local odd
basis for this ring is $\{c^a_\La\}$, $0\leq |\La|\leq r$. Let
$\cS^*_r[F;Y]$ be the differential graded algebra (henceforth DGA)
of graded differential forms on the graded manifold
$(J^rY,\cA_r)$. The inverse system of jet manifolds $J^{r-1}Y
\leftarrow J^rY$ yields the direct system of DGAs
\be
\cS^*[F;Y]\ar \cS^*_1[F;Y]\ar\cdots \cS^*_r[F;Y]\ar\cdots.
\ee
Its direct limit $\cS^*_\infty[F;Y]$ is the DGA of all graded
differential forms on graded manifolds $(J^rY,\cA_r)$. Recall the
formulas
\be
\f\w\f' =(-1)^{|\f||\f'| +[\f][\f']}\f'\w \f, \qquad d(\f\w\f')=
d\f\w\f' +(-1)^{|\f|}\f\w d\f,
\ee
where $[\f]$ denotes the Grassmann parity. The DGA
$S^*_\infty[F;Y]$ contains the subalgebra $\cO^*_\infty Y$ of all
exterior forms on jet manifolds $J^rY$. It is an $\cO^0_\infty
Y$-algebra locally generated by elements
$(c^a_\La,dx^\la,dy^i_\La, dc^a_\La)$, $0\leq |\La|$. The
collective symbol $(s^A)$ further stands for the tuple
$(y^i,c^a)$, called the local basis for the DGA
$\cS^*_\infty[F;Y]$. We denote $[A]=[s^A]=[s^A_\La]$.

The DGA $\cS^*_\infty[F;Y]$ is decomposed into the
Grassmann-graded variational bicomplex of modules
$\cS^{k,r}_\infty[F;Y]$ of $r$-horizontal and $k$-contact graded
forms locally generated by one-forms $dx^\la$ and
$\th^A_\La=ds_\La^A -s^A_{\la +\La} dx^\la$. It contains the
variational subcomplex
\be
0\to \Bbb R\ar \cS^0_\infty[F;Y]\ar^{d_H}\cS^{0,1}_\infty[F;Y]
\cdots \ar^{d_H} \cS^{0,n}_\infty[F;Y]\ar^\dl
\cS^{1,n}_\infty[F;Y],
\ee
where
\be
d_H(\f)=dx^\la\w d_\la\f, \qquad d_\la =\dr_\la +
\op\sum_{0\leq|\La|} s_{\la\La}^A\dr^\La_A,
\ee
is the total differential and
\be
\dl L= \th^A\w \cE_A d^nx=\op\sum_{0\leq|\La|} (-1)^{|\La|}\th^A\w
d_\La (\dr^\La_A \cL) d^nx, \qquad
 d_\La=d_{\la_1}\cdots d_{\la_k},
\ee
is the variational operator. Lagrangians and Euler--Lagrange
operators are defined as even elements
\mar{0709}\beq
L=\cL d^nx\in \cS^{0,n}_\infty[F;Y], \qquad \dl L= \th^A\w \cE_A
d^nx\in \cS^{1,n}_\infty[F;Y]. \label{0709}
\eeq
Further, we call a pair $(\cS^*_\infty[F;Y],L)$ the Lagrangian
field theory.

Cohomology of the variational bicomplex has been
obtained.$^{2,12}$ Let us mention the following results.

\begin{prop} \mar{811} \label{811}
Any variationally trivial (i.e., $\dl$-closed) graded density
$L\in \cS^{0,n}_\infty[F;Y]$ takes the form $L=d_H\psi + h_0\vf$,
where $\vf$ is a closed exterior $n$-form on $Y$ where
$h_0(dy^i)=y^i_\m dx^\m$. In particular, any odd variationally
trivial graded density is $d_H$-exact.
\end{prop}

\begin{prop} \mar{811'} \label{811'} The form $dL-\dl L$ is $d_H$-exact for any graded
density $L\in \cS^{0,n}_\infty[F;Y]$.
\end{prop}

In order to treat symmetries of Lagrangian field theory
$(\cS^*_\infty[F;Y],L)$ in a very general setting, we consider
graded derivations of the $\Bbb R$-ring $\cS^0_\infty[F;Y]$.$^2$
They take the form
\mar{809}\beq
 \vt=\vt^\la\dr_\la + \op\sum_{0\leq|\La|} \vt_\La^A\dr^\La_A,
 \qquad \dr^\La_A(s_\Si^B)=\dr^\La_A\rfloor
ds_\Si^B=\dl_A^B\dl^\La_\Si. \label{809}
\eeq
Any graded derivation $\vt$ (\ref{809}) yields the Lie derivative
\be
\bL_\vt\f=\vt\rfloor d\f+ d(\vt\rfloor\f)
\ee
of the DGA $\cS^*_\infty[F;Y]$ which obeys the relations
\be
\bL_\vt\f=\vt\rfloor d\f+ d(\vt\rfloor\f), \qquad
\bL_\vt(\f\w\f')=\bL_\vt(\f)\w\f'
+(-1)^{[\vt][\f]}\f\w\bL_\vt(\f').
\ee
A graded derivation $\vt$ (\ref{809}) is called contact if the Lie
derivative $\bL_\vt$ preserves the ideal of contact graded forms
of the DGA $\cS^*_\infty[F;Y]$. Any contact graded derivation
admits the decomposition
\mar{810}\beq
\vt=\vt_H+\vt_V=\vt^\la d_\la + (\vt^A\dr_A +\op\sum_{0<|\La|}
d_\La\vt^A\dr_A^\La) \label{810}
\eeq
into the horizontal and vertical parts $\vt_H$ and $\vt_V$.

Given a graded density $L$ (\ref{0709}), a contact graded
derivation $\vt$ (\ref{810}) is said to be its variational
symmetry if the Lie derivative $\bL_\vt L$ of $L$ is a
variationally trivial graded density. If $\bL_\vt L=0$, a
variational symmetry of $L$ is called its exact symmetry.

\begin{prop} \mar{812} \label{812}
A contact graded derivation $\vt$ (\ref{810}) is a variational
symmetry iff its vertical part $\vt_V$  is also.$^2$
\end{prop}

Therefore, we further restrict our consideration to vertical
contact graded derivations $\vt$. Such a derivation is the jet
prolongation
\mar{0672}\beq
\vt=\up^A\dr_A + \op\sum_{0<|\La|} d_\La\up^A\dr_A^\La
\label{0672}
\eeq
of its restriction $\up=\up^A\dr_A$ to the ring
$\cS^0_\infty[F;Y]$. Therefore, the relations
\mar{s2,3}\ben
&& \vt\rfloor d_H\f=-d_H(\vt\rfloor\f), \qquad
\f\in\cS^*_\infty[F;Y], \label{s2}\\
&& \bL_\vt(d_H\f)=d_H(\bL_\vt\f) \label{s3}
\een
hold. By virtue of the relation (\ref{s2}) and Proposition
\ref{811'}, the Lie derivative $\bL_{\vt}L$ of any graded density
$L$ admits the decomposition
\mar{xx10}\beq
\bL_{\vt}L= \vt\rfloor dL=\vt\rfloor\dl L + \vt\rfloor (dL-\dl L)
=\up\rfloor\dl L +d_H\si=\up^A\cE_A d^nx+d_H\si, \label{xx10}
\eeq
called the first variational formula. A glance at the expression
(\ref{xx10}) shows that $\vt$ (\ref{0672}) is a variational
symmetry of $L$ iff the graded density $\up\rfloor\dl L$ is
variationally trivial.

By virtue of the relation (\ref{s3}), any graded derivation
(\ref{0672}) is a variational symmetry of a variationally trivial
graded density. It follows that variational symmetries of a graded
density constitute a real Lie algebra $\cG_L$.

A graded derivation $\vt$ (\ref{0672}) is called nilpotent if
$\bL_\vt(\bL_\vt\f)=0$ for any horizontal form $\f\in
\cS^{0,*}_\infty[F;Y]$. One can show that $\vt$ (\ref{0672}) is
nilpotent only if it is odd and iff $\vt(\up)=0$.$^2$

For the sake of simplicity, the common symbol $\up$ further stands
for the graded derivation $\vt$ (\ref{0672}), its first term
$\up$, and the Lie derivative $\bL_\vt$. We agree to call $\up$
the graded derivation of the DGA $\cS^*_\infty[F;Y]$. Its right
graded derivations $\op\up^\lto ={\op\dr^\lto}_A\up^A$ are also
considered.

\bigskip
\bigskip

\noindent {\bf III. NOETHER IDENTITIES}
\bigskip

To describe reducible NI of Lagrangian field theory
$(\cS^*_\infty[F;Y],L)$,$^{8,11}$ let us introduce the following
notation. Given a vector bundle $E\to X$, we call
\be
\ol E=E^*\ot\op\w^n T^*X
\ee
the density-dual of $E$. The density-dual of a graded vector
bundle $E=E^0\oplus E^1$ is $\ol E=\ol E^1\oplus \ol E^0$. Given a
graded vector bundle $E=E^0\oplus E^1$ over $Y$, we consider the
composite bundle $E\to E^0\to X$ and denote
$\cP^*_\infty[E;Y]=\cS^*_\infty[E;E^0]$. Let $VF$ be the vertical
tangent bundle of $F\to X$, the density-dual of the vector bundle
$VF\to F$ is
\be
\ol{VF}=V^*F\op\ot_F\op\w^n T^*X.
\ee

Let us enlarge $\cS^*_\infty[F;Y]$ to the DGA
$\cP^*_\infty[\ol{VF};Y]$ with the local basis $(s^A, \ol s_A)$,
$[\ol s_A]=([A]+1){\rm mod}\,2$. Its elements $\ol s_A$ are called
antifields of antifield number Ant$[\ol s_A]= 1$. The DGA
$\cP^*_\infty[\ol{VF};Y]$ is endowed with the nilpotent right
graded derivation $\ol\dl=\rdr^A \cE_A$, where $\cE_A$ are the
variational derivatives (\ref{0709}). Then we have the chain
complex
\mar{v042}\beq
0\lto \im\ol\dl \llr^{\ol\dl} \cP^{0,n}_\infty[\ol{VF};Y]_1
\llr^{\ol\dl} \cP^{0,n}_\infty[\ol{VF};Y]_2 \label{v042}
\eeq
of graded densities of antifield number $\leq 2$. Its one-cycles
\mar{0712}\beq
\ol\dl \Phi=0, \qquad \Phi= \op\sum_{0\leq|\La|} \Phi^{A,\La}\ol
s_{\La A} d^nx \in \cP^{0,n}_\infty[\ol{VF};Y]_1,\label{0712}
\eeq
define NI of Lagrangian field theory $(\cS^*_\infty[F;Y],L)$. In
particular, one-chains $\Phi \in \cP^{0,n}_\infty[\ol{VF};Y]_1$
are necessarily NI if they are boundaries. Therefore, these NI are
called trivial. Accordingly, non-trivial NI modulo the trivial
ones are associated to elements of the first homology
$H_1(\ol\dl)$ of the complex (\ref{v042}).

Non-trivial NI obey first-stage NI. To describe them, let us
assume that the module $H_1(\ol \dl)$ is finitely generated.
Namely, there exists a projective $C^\infty(X)$-module
$\cC_{(0)}\subset H_1(\ol \dl)$ of finite rank possessing the
local basis $\{\Delta_r\}$ such that any element $\Phi\in H_1(\ol
\dl)$ factorizes as
\mar{xx2}\beq
\Phi= \op\sum_{0\leq|\Xi|} G^{r,\Xi} d_\Xi \Delta_r d^nx, \qquad
\Delta_r=\op\sum_{0\leq|\La|} \Delta_r^{A,\La}\ol s_{\La A},\qquad
G^{r,\Xi},\Delta_r^{A,\La}\in \cS^0_\infty[F;Y], \label{xx2}
\eeq
through elements of $\cC_{(0)}$. Thus, all non-trivial NI
(\ref{0712}) result from the NI
\mar{v64}\beq
\ol\dl\Delta_r= \op\sum_{0\leq|\La|} \Delta_r^{A,\La} d_\La
\cE_A=0, \label{v64}
\eeq
called the complete NI. By virtue of the Serre--Swan
theorem,$^{11}$ the module $\cC_{(0)}$ is isomorphic to the
$C^\infty(X)$-module of sections of the density-dual $\ol E_0$ of
some graded vector bundle $E_0\to X$. Let us enlarge
$\cP^*_\infty[\ol{VF};Y]$ to the DGA
\be
\ol\cP^*_\infty\{0\}=\cP^*_\infty[\ol{VF}\oplus_Y \ol E_0;Y]
\ee
possessing the local basis $(s^A,\ol s_A, \ol c_r)$ of Grassmann
parity $[\ol c_r]=([\Delta_r]+1){\rm mod}\,2$ and of antifield
number ${\rm Ant}[\ol c_r]=2$. This DGA is provided with the odd
right graded derivation $\dl_0=\ol\dl + \rdr^r\Delta_r$ which is
nilpotent iff the NI (\ref{v64}) hold. Then we have the chain
complex
\mar{v66}\beq
0\lto \im\ol\dl \op\lto^{\ol\dl}
\cP^{0,n}_\infty[\ol{VF};Y]_1\op\lto^{\dl_0}
\ol\cP^{0,n}_\infty\{0\}_2 \op\lto^{\dl_0}
\ol\cP^{0,n}_\infty\{0\}_3 \label{v66}
\eeq
of graded densities of antifield number $\leq 3$. It possesses
trivial homology $H_0(\dl_0)$ and $H_1(\dl_0)$. Its two-cycles
define the first-stage NI
\mar{v79}\ben
&& \dl_0 \Phi=0, \qquad \Phi= G + H= \op\sum_{0\leq|\La|} G^{r,\La}\ol c_{\La r}d^nx +
\op\sum_{0\leq|\La|,|\Si|} H^{(A,\La)(B,\Si)}\ol s_{\La A}\ol
s_{\Si B}d^nx,  \nonumber\\
&& \op\sum_{0\leq|\La|} G^{r,\La}d_\La\Delta_r d^nx =-\ol\dl H.
\label{v79}
\een
However, the converse need not be true. One can show that NI
(\ref{v79}) are cycles iff any $\ol\dl$-cycle $\Phi\in
\cP^{0,n}_\infty[\ol{VF};Y]_2$ is a $\dl_0$-boundary.$^{11}$ Any
boundary $\Phi\in \ol\cP^{0,n}_\infty\{0\}_2$ necessarily defines
first-stage NI (\ref{v79}), called trivial. Accordingly,
non-trivial first-stage NI modulo the trivial ones are identified
with elements of the second homology $H_2(\dl_0)$ of the complex
(\ref{v66}).

Non-trivial first-stage NI obey second-stage NI, and so on.
Iterating the arguments, one can characterize $N$-stage reducible
Lagrangian field theory $(\cS^*_\infty[F;Y],L)$ as follows. There
are graded vector bundles $E_0,\ldots, E_N$ over $X$, and the DGA
$\cP^*_\infty[\ol{VF};Y]$ is enlarged to the DGA
\mar{v91}\beq
\ol\cP^*_\infty\{N\}=\cP^*_\infty[\ol{VF}\op\oplus_Y \ol
E_0\op\oplus_Y\cdots\op\oplus_Y \ol E_N;Y] \label{v91}
\eeq
with the local basis $(s^A,\ol s_A, \ol c_r, \ol c_{r_1}, \ldots,
\ol c_{r_N})$ of antifield number Ant$[\ol c_{r_k}]=k+2$. The DGA
(\ref{v91}) is provided with the nilpotent right graded derivation
\mar{v92,'}\ben
&&\dl_{KT}=\rdr^A\cE_A +
\op\sum_{0\leq|\La|}\rdr^r\Delta_r^{A,\La}\ol s_{\La A} +
\op\sum_{1\leq k\leq N}\rdr^{r_k} \Delta_{r_k},
\label{v92}\\
&& \Delta_{r_k}= \op\sum_{0\leq|\La|}
\Delta_{r_k}^{r_{k-1},\La}\ol c_{\La r_{k-1}} + \op\sum_{0\leq
|\Si|, |\Xi|}(h_{r_k}^{(r_{k-2},\Si)(A,\Xi)}\ol c_{\Si r_{k-2}}\ol
s_{\Xi A}+...), \label{v92'}
\een
of antifield number -1, where the index $k=-1$ stands for $\ol
s_A$. It is called the KT operator. With this graded derivation,
the module $\ol\cP^{0,n}_\infty\{N\}_{\leq N+3}$ of densities of
antifield number $\leq (N+3)$ is decomposed into the exact KT
chain complex
\mar{v94}\ben
&& 0\lto \im \ol\dl \llr^{\ol\dl}
\cP^{0,n}_\infty[\ol{VF};Y]_1\llr^{\dl_0}
\ol\cP^{0,n}_\infty\{0\}_2\llr^{\dl_1}
\ol\cP^{0,n}_\infty\{1\}_3\cdots
\label{v94}\\
&& \qquad
 \llr^{\dl_{N-1}} \ol\cP^{0,n}_\infty\{N-1\}_{N+1}
\llr^{\dl_{KT}} \ol\cP^{0,n}_\infty\{N\}_{N+2}\llr^{\dl_{KT}}
\ol\cP^{0,n}_\infty\{N\}_{N+3} \nonumber
\een
which satisfies the following homology condition.

\begin{cond} \label{v155} \mar{v155} Any $\dl_{k<N}$-cycle
$\f\in \ol\cP_\infty^{0,n}\{k\}_{k+3}\subset
\ol\cP_\infty^{0,n}\{k+1\}_{k+3}$ is a $\dl_{k+1}$-boundary.
\end{cond}

Given the KT complex (\ref{v94}), the nilpotentness $\dl_{KT}^2=0$
of its boundary operator (\ref{v92}) is equivalent to the complete
non-trivial NI (\ref{v64}) and the complete non-trivial $(1\leq
k\leq N)$-stage NI
\mar{v93}\beq
\op\sum_{0\leq|\La|} \Delta_{r_k}^{r_{k-1},\La}d_\La
(\op\sum_{0\leq|\Si|} \Delta_{r_{k-1}}^{r_{k-2},\Si}\ol c_{\Si
r_{k-2}}) = -  \ol\dl(\op\sum_{0\leq |\Si|,
|\Xi|}h_{r_k}^{(r_{k-2},\Si)(A,\Xi)}\ol c_{\Si r_{k-2}}\ol s_{\Xi
A}). \label{v93}
\eeq

\bigskip
\bigskip

\noindent {\bf IV. GAUGE SYMMETRIES}
\bigskip

We define non-trivial gauge and higher-stage gauge symmetries of
Lagrangian field theory $(\cS^*_\infty[F;Y],L)$ as those
associated to the NI (\ref{v64}) and higher-stage NI (\ref{v93})
by means of the inverse second Noether theorem.

Let us start with the following notation. Given the DGA
$\ol\cP^*_\infty\{N\}$ (\ref{v91}), we consider the DGA
\mar{w5}\beq
\cP^*_\infty\{N\}=\cP^*_\infty[F\op\oplus_Y E_0\op\oplus_Y\cdots
\op\oplus_Y E_N;Y], \label{w5}
\eeq
possessing the local basis $(s^A, c^r, c^{r_1}, \ldots, c^{r_N})$,
$[c^{r_k}]=([\ol c_{r_k}]+1){\rm mod}\,2$, and the DGA
\mar{w6}\beq
P^*_\infty\{N\}=\cP^*_\infty[\ol{VF}\op\oplus_Y E_0\oplus\cdots
\op\oplus_Y E_N \op\oplus_Y \ol E_0\op\oplus_Y\cdots\op\oplus_Y
\ol E_N;Y]
 \label{w6}
\eeq
with the local basis $(s^A, \ol s^A, c^r, c^{r_1}, \ldots,
c^{r_N},\ol c_r, \ol c_{r_1}, \ldots, \ol c_{r_N})$. Their
elements $c^{r_k}$ are called $k$-stage ghosts of ghost number
gh$[c^{r_k}]=k+1$ and antifield number ${\rm
Ant}[c^{r_k}]=-(k+1)$. The $C^\infty(X)$-module $\cC^{(k)}$ of
$k$-stage ghosts is the density dual of the module $\cC_{(k+1})$
of $(k+1)$-stage antifields. The DGAs $\ol\cP^*_\infty\{N\}$
(\ref{v91}) and $\cP^*_\infty\{N\}$ (\ref{w5}) are subalgebras of
$P^*_\infty\{N\}$ (\ref{w6}). The KT operator $\dl_{KT}$
(\ref{v92}) is naturally extended to a graded derivation of the
DGA $P^*_\infty\{N\}$.

We refer to the following formulas in the sequel.$^3$ Any graded
form $\f\in \cS^*_\infty[F;Y]$ and any finite tuple $(f^\La)$,
$0\leq |\La|$, of local graded functions $f^\La\in
\cS^0_\infty[F;Y]$ obey the relations
\mar{qq1}\ben
&& \op\sum_{0\leq |\La|\leq k} f^\La d_\La \f\w d^nx= \op\sum_{0\leq
|\La|}(-1)^{|\La|}d_\La (f^\La)\f\w d^nx +d_H\si,
\label{qq1a}\\
&& \op\sum_{0\leq |\La|\leq k} (-1)^{|\La|}d_\La(f^\La \f)=
\op\sum_{0\leq |\La|\leq k} \eta (f)^\La d_\La \f, \label{qq1b}\\
&& \eta (f)^\La = \op\sum_{0\leq|\Si|\leq k-|\La|}(-1)^{|\Si+\La|}
\frac{(|\Si+\La|)!}{|\Si|!|\La|!} d_\Si f^{\Si+\La},
\label{qq1c}\\
&& \eta(\eta(f))^\La=f^\La. \label{qq1d}
\een

\begin{theo} \label{w35} \mar{w35} Given the KT complex (\ref{v94}),
the module of graded densities $\cP_\infty^{0,n}\{N\}$ is
decomposed into the cochain sequence
\mar{w108,'}\ben
&& 0\to \cS^{0,n}_\infty[F;Y]\ar^{\bu}
\cP^{0,n}_\infty\{N\}^1\ar^{\bu}
\cP^{0,n}_\infty\{N\}^2\ar^{\bu}\cdots, \label{w108}\\
&& \bu=u + u^{(1)}+\cdots + u^{(N)}=u^A\frac{\dr}{\dr s^A} +
u^r\frac{\dr}{\dr c^r} +\cdots + u^{r_{N-1}}\frac{\dr}{\dr
c^{r_{N-1}}}, \label{w108'}
\een
graded in ghost number. Its ascent operator $\bu$ (\ref{w108'}) is
an odd graded derivation of ghost number 1 where $u$ (\ref{w33})
is a variational symmetry of a Lagrangian $L$ and the graded
derivations $u_{(k)}$ (\ref{w38}), $k=1,\ldots, N$, obey the
relations (\ref{w34}).
\end{theo}

\begin{proof} Given the KT operator (\ref{v92}), let us extend an original
Lagrangian $L$ to the Lagrangian
\mar{w8}\beq
L_e=L+L_1=L + \op\sum_{0\leq k\leq N} c^{r_k}\Delta_{r_k}d^nx=L
+\dl_{KT}( \op\sum_{0\leq k\leq N} c^{r_k}\ol c_{r_k}d^nx)
\label{w8}
\eeq
of zero antifield number. It is readily observed that the KT
operator $\dl_{KT}$ is a variational symmetry of $L_e$. Since
$\dl_{KT}$ is odd, it follows from the first variational formula
(\ref{xx10}) and Proposition \ref{811} that
\mar{w16}\ben
&& [\frac{\op\dl^\lto \cL_e}{\dl \ol s_A}\cE_A
+\op\sum_{0\leq k\leq N} \frac{\op\dl^\lto \cL_e}{\dl \ol
c_{r_k}}\Delta_{r_k}]d^nx = [\up^A\cE_A + \op\sum_{0\leq k\leq
N}\up^{r_k}\frac{\dl
\cL_e}{\dl c^{r_k}}]d^nx= d_H\si,  \label{w16}\\
&& \up^A= \frac{\op\dl^\lto \cL_e}{\dl \ol s_A}=u^A+w^A
=\op\sum_{0\leq|\La|} c^r_\La\eta(\Delta^A_r)^\La +
 \op\sum_{1\leq i\leq N}\op\sum_{0\leq|\La|}
c^{r_i}_\La\eta(\op\dr^\lto{}^A(h_{r_i}))^\La, \nonumber\\
&& \up^{r_k}=\frac{\op\dl^\lto \cL_e}{\dl \ol c_{r_k}} =u^{r_k}+
w^{r_k}= \op\sum_{0\leq|\La|}
c^{r_{k+1}}_\La\eta(\Delta^{r_k}_{r_{k+1}})^\La +
\op\sum_{k+1<i\leq N} \op\sum_{0\leq|\La|}
c^{r_i}_\La\eta(\op\dr^\lto{}^{r_k}(h_{r_i}))^\La. \nonumber
\een
The equality (\ref{w16}) falls into the set of equalities
\mar{w19,20}\ben
&& \frac{\op\dl^\lto (c^r\Delta_r)}{\dl \ol s_A}\cE_A d^nx
=u^A\cE_A d^nx=d_H\si_0, \label{w19}\\
&&  [\frac{\op\dl^\lto (c^{r_k}\Delta_{r_k})}{\dl \ol s_A}\cE_A
+\op\sum_{0\leq i<k} \frac{\op\dl^\lto (c^{r_k}\Delta_{r_k})}{\dl
\ol c_{r_i}}\Delta_{r_i}] d^nx= d_H\si_k, \qquad k=1,\ldots,N.
\label{w20}
\een
A glance at the equality (\ref{w19}) shows that, by virtue of the
first variational formula (\ref{xx10}), the odd graded derivation
\mar{w33}\beq
u= u^A\frac{\dr}{\dr s^A}, \qquad u^A =\op\sum_{0\leq|\La|}
c^r_\La\eta(\Delta^A_r)^\La, \label{w33}
\eeq
of $\cP^0\{0\}$ is a variational symmetry of a Lagrangian $L$.
Every equality (\ref{w20}) falls into a set of equalities graded
by the polynomial degree in antifields. Let us consider that of
them linear in antifields $\ol c_{r_{k-2}}$. We have
\be
&& [\frac{\op\dl^\lto}{\dl \ol
s_A}(c^{r_k}\op\sum_{0\leq|\Si|,|\Xi|}h_{r_k}^{(r_{k-2},\Si)(A,\Xi)}
\ol
c_{\Si r_{k-2}}\ol s_{\Xi A})\cE_A + \\
&& \qquad \frac{\op\dl^\lto}{\dl \ol
c_{r_{k-1}}}(c^{r_k}\op\sum_{0\leq|\Si|}\Delta_{r_k}^{r'_{k-1},\Si}\ol
c_{\Si r'_{k-1}})\op\sum_{0\leq|\Xi|}
\Delta_{r_{k-1}}^{r_{k-2},\Xi}\ol c_{\Xi r_{k-2}}]d^nx= d_H\si_k.
\ee
This equality is brought into the form
\be
 [\op\sum_{0\leq|\Xi|}
(-1)^{|\Xi|}d_\Xi(c^{r_k}\op\sum_{0\leq|\Si|}
h_{r_k}^{(r_{k-2},\Si)(A,\Xi)} \ol c_{\Si r_{k-2}})\cE_A +
u^{r_{k-1}}\op\sum_{0\leq|\Xi|} \Delta_{r_{k-1}}^{r_{k-2},\Xi}\ol
c_{\Xi r_{k-2}}] d^nx= d_H\si_k.
\ee
Using the relation (\ref{qq1a}), we obtain the equality
\be
[\op\sum_{0\leq|\Xi|} c^{r_k}\op\sum_{0\leq|\Si|}
h_{r_k}^{(r_{k-2},\Si)(A,\Xi)} \ol c_{\Si r_{k-2}} d_\Xi\cE_A +
u^{r_{k-1}}\op\sum_{0\leq|\Xi|} \Delta_{r_{k-1}}^{r_{k-2},\Xi}\ol
c_{\Xi r_{k-2}}]d^nx= d_H\si'_k.
\ee
The variational derivative of both its sides with respect to $\ol
c_{r_{k-2}}$ leads to the relation
\mar{w34}\ben
&&\op\sum_{0\leq|\Si|} d_\Si u^{r_{k-1}}\frac{\dr}{\dr
c^{r_{k-1}}_\Si} u^{r_{k-2}} =\ol\dl(\al^{r_{k-2}}),\label{w34}\\
&& \al^{r_{k-2}} = -\op\sum_{0\leq|\Si|}
\eta(h_{r_k}^{(r_{k-2})(A,\Xi)})^\Si d_\Si(c^{r_k} \ol s_{\Xi A}),
\nonumber
\een
which the odd graded derivation
\mar{w38}\beq
u^{(k)}= u^{r_{k-1}}\frac{\dr}{\dr c^{r_{k-1}}}, \qquad
u^{r_{k-1}}=\op\sum_{0\leq|\La|}
c^{r_k}_\La\eta(\Delta^{r_{k-1}}_{r_k})^\La, \qquad k=1,\ldots,N,
\label{w38}
\eeq
satisfies. Graded derivations $u$ (\ref{w33}) and $u^{(k)}$
(\ref{w38}) are assembled into the ascent operator $\bu$
(\ref{w108'}) of the cochain sequence (\ref{w108}).
\end{proof}

A glance at the expression (\ref{w33}) shows that the variational
symmetry $u$ is a linear differential operator on the
$C^\infty(X)$-module $\cC^{(0)}$ of ghosts with values in the real
space $\cG_L$ of variational symmetries. Following Definition
\ref{s7} extended to Lagrangian theories of odd fields, we call
$u$ (\ref{w33}) the gauge symmetry of a Lagrangian $L$ which is
associated to the NI (\ref{v64}). This association is unique due
to the following.

\begin{prop} \mar{825} \label{825}
The variational derivative of the equality (\ref{w19}) with
respect to ghosts $c^r$ leads to the equality
\be
\dl_r(u^A\cE_A
d^nx)=\op\sum_{0\leq|\La|}(-1)^{|\La|}d_\La(\eta(\Delta^A_r)^\La\cE_A)=
\op\sum_{0\leq|\La|}(-1)^{|\La|} \eta(\eta(\Delta^A_r))^\La
d_\La\cE_A=0,
\ee
which reproduces the complete non-trivial NI (\ref{v64}) by means
of the relation (\ref{qq1d}).
\end{prop}

Moreover, the gauge symmetry $u$ (\ref{w33}) is complete in the
following sense. Let
\be
\op\sum_{0\leq|\Xi|} C^RG^{r,\Xi}_R d_\Xi \Delta_r d^nx
\ee
be some projective $C^\infty(X)$-module of finite rank of
non-trivial NI (\ref{xx2}) parameterized by the corresponding
ghosts $C^R$. We have the equalities
\be
&& 0=\op\sum_{0\leq|\Xi|} C^RG^{r,\Xi}_R d_\Xi
(\op\sum_{0\leq|\La|}\Delta_r^{A,\La}d_\La \cE_A) d^nx=\\
&& \qquad
\op\sum_{0\leq|\La|}(\op\sum_{0\leq|\Xi|}\eta(G^r_R)^\Xi C^R_\Xi)
\Delta_r^{A,\La}d_\La \cE_Ad^nx+d_H(\si)=\\
&& \qquad
\op\sum_{0\leq|\La|}(-1)^{|\La|}d_\La(\Delta_r^{A,\La}\op\sum_{0\leq|\Xi|}\eta(G^r_R)^\Xi
C^R_\Xi)\cE_A d^nx +d_H\si =\\
&& \qquad
\op\sum_{0\leq|\La|}\eta(\Delta_r^A)^\La
d_\La(\op\sum_{0\leq|\Xi|}\eta(G^r_R)^\Xi C^R_\Xi)\cE_A d^nx
+d_H\si=\\
&&\qquad \op\sum_{0\leq|\La|}u_r^{A,\La}d_\La(\op\sum_{0\leq|\Xi|}\eta(G^r_R)^\Xi
C^R_\Xi)\cE_A d^nx +d_H\si.
\ee
It follows that the graded derivation
\be
d_\La(\op\sum_{0\leq|\Xi|}\eta(G^r_R)^\Xi
C^R_\Xi)u_r^{A,\La}\frac{\dr}{\dr s^A}
\ee
is a variational symmetry of a Lagrangian $L$ and, consequently,
its gauge symmetry parameterized by ghosts $C^R$. It factorizes
through the gauge symmetry (\ref{w33}) by putting ghosts
\be
c^r= \op\sum_{0\leq|\Xi|}\eta(G^r_R)^\Xi C^R_\Xi.
\ee
Thus, we come to the following definition.

\begin{defi} \mar{s4} \label{s4}
The odd graded derivation $u$ (\ref{w33}) is said to be a complete
non-trivial gauge symmetry of Lagrangian field theory associated
to complete non-trivial NI (\ref{v64}).
\end{defi}

For instance, if a complete non-trivial gauge symmetry (\ref{w33})
is of second jet order in ghosts, i.e.,
\mar{0656}\beq
u=(c^r u_r^A +c^r_\m u^{A,\m}_r + c^r_{\nu\m} u_r^{A,\nu\m})\dr_A,
\label{0656}
\eeq
the corresponding NI (\ref{v64}) take the form
\mar{0657}\beq
u^A_r\cE_A - d_\m(u^{A,\m}_r\cE_A) + d_{\nu\m}(u_r^{A,\nu\m}
\cE_A)=0. \label{0657}
\eeq

Turn now to the relation (\ref{w34}). For $k=1$, it takes the form
\be
\op\sum_{0\leq|\Si|} d_\Si u^r\frac{\dr}{\dr c^r_\Si} u^A =\ol
\dl(\al^A)
\ee
of a first-stage gauge symmetry condition on-shell which the
non-trivial gauge symmetry $u$ (\ref{w33}) satisfies. Therefore,
one can treat the odd graded derivation
\be
u^{(1)}= u^r\frac{\dr}{\dr c^r}, \qquad u^r=\op\sum_{0\leq|\La|}
c^{r_1}_\La\eta(\Delta^r_{r_1})^\La,
\ee
as a first-stage gauge symmetry associated to the complete
non-trivial first-stage NI
\be
 \op\sum_{0\leq|\La|} \Delta_{r_1}^{r,\La}d_\La
(\op\sum_{0\leq|\Si|} \Delta_r^{A,\Si}\ol s_{\Si A}) = -
\ol\dl(\op\sum_{0\leq |\Si|, |\Xi|}h_{r_1}^{(B,\Si)(A,\Xi)}\ol
s_{\Si B}\ol s_{\Xi A}).
\ee

Iterating the arguments, one comes to the relation (\ref{w34})
which provides a $k$-stage gauge symmetry condition which is
associated to the complete non-trivial $k$-stage NI (\ref{v93}).

\begin{prop} \mar{826} \label{826}
Conversely, given the $k$-stage gauge symmetry condition
(\ref{w34}), the $k$-stage NI (\ref{v93}) are reproduced.$^3$
\end{prop}

Accordingly, we call the odd graded derivation $u_{(k)}$
(\ref{w38}) the $k$-stage gauge symmetry. It is complete as
follows.$^3$ Let
\be
\op\sum_{0\leq|\Xi|} C^{R_k}G^{r_k,\Xi}_{R_k} d_\Xi \Delta_{r_k}
d^nx
\ee
be a projective $C^\infty(X)$-module of finite rank of non-trivial
$k$-stage NI (\ref{xx2}) factorizing through the complete ones
(\ref{v92'}), and which are parameterized by the corresponding
ghosts $C^{R_k}$. One can show that it defines a $k$-stage gauge
symmetry factorizing through $u^{(k)}$ (\ref{w38}) by putting
$k$-stage ghosts
\be
c^{r_k}= \op\sum_{0\leq|\Xi|}\eta(G^{r_k}_{R_k})^\Xi C^{R_k}_\Xi.
\ee

\begin{defi} \mar{s6} \label{s6}
The odd graded derivation $u_{(k)}$ (\ref{w38}) is said to be a
complete non-trivial $k$-stage gauge symmetry of a Lagrangian $L$.
\end{defi}

In accordance with Definitions \ref{s4} and \ref{s6}, components
of the ascent operator $\bu$ (\ref{w108'}) are complete
non-trivial gauge and higher-stage gauge symmetries. Therefore, we
agree to call this operator the gauge operator. In these terms,
Theorem \ref{w35} is the inverse second Noether theorem. The
corresponding direct second Noether theorem is stated by
Propositions \ref{825} and \ref{826}.

\bigskip
\bigskip

\noindent {\bf V. ALGEBRA OF GAUGE SYMMETRIES}
\bigskip

In contrast with the KT operator (\ref{v92}), the gauge operator
$\bu$ (\ref{w108}) need not be nilpotent. Following the example of
Yang--Mills gauge theory, let us study its extension to a
nilpotent graded derivation
\mar{w109}\ben
&& \bb=\bu+ \g=\bu + \op\sum_{1\leq k\leq N+1}\g^{(k)}=
\bu + \op\sum_{1\leq k\leq N+1}\g^{r_{k-1}}\frac{\dr}{\dr
c^{r_{k-1}}}= \label{w109} \\
&& \qquad (u^A\frac{\dr}{\dr s^A}+ \g^r\frac{\dr}{\dr c^r}) +
\op\sum_{0\leq k\leq N-1} (u^{r_k}\frac{\dr}{\dr c^{r_k}}+
\g^{r_{k+1}}\frac{\dr}{\dr c^{r_{k+1}}}) \nonumber
\een
of ghost number 1 by means of antifield-free terms $\g^{(k)}$ of
higher polynomial degree in ghosts $c^{r_i}$ and their jets
$c^{r_i}_\La$, $0\leq i<k$. We call $\bb$ (\ref{w109}) the BRST
operator. The following necessary condition holds.

\begin{prop} \label{lmp6} \mar{lmp6} The gauge operator
(\ref{w108}) admits the nilpotent extension (\ref{w109}) only if
the gauge symmetry conditions (\ref{w34}) and the higher-stage NI
(\ref{v93}) are satisfied off-shell.
\end{prop}

\begin{proof}
It is easily justified that, if the graded derivation $\bb$
(\ref{w109}) is nilpotent, then  the right hand sides of the
equalities (\ref{w34}) equal zero, i.e.,
\mar{850}\beq
u^{(k+1)}(u^{(k)})=0, \qquad 0\leq k\leq N-1, \qquad u^{(0)}=u.
\label{850}
\eeq
Using the relations (\ref{qq1a}) -- (\ref{qq1c}), one can show
that, in this case, the right hand sides of the higher-stage NI
(\ref{v93}) also equal zero.$^4$  It follows that the summand
$G_{r_k}$ of each cocycle $\Delta_{r_k}$ (\ref{v92'}) is
$\dl_{k-1}$-closed. Then its summand $h_{r_k}$ is also
$\dl_{k-1}$-closed and, consequently, $\dl_{k-2}$-closed. Hence it
is $\dl_{k-1}$-exact by virtue of Condition \ref{v155}. Therefore,
$\Delta_{r_k}$ contains only the term $G_{r_k}$ linear in
antifields.
\end{proof}

It follows at once from the equalities (\ref{850}) that the gauge
operator of higher-stage gauge symmetries
\be
u_{HS}=\bu-u=u^{(1)}+\cdots + u^{(N)}
\ee
is nilpotent, and $\bu(\bu)=u(\bu)$. Therefore, the nilpotency
condition for the BRST operator $\bb$ (\ref{w109}) takes the form
\mar{851}\beq
\bb(\bb)=(u+\g)(\bu) +(u+u_{HS}+\g)(\g)=0. \label{851}
\eeq
Let us denote
\be
&& \g^{(0)}=0, \qquad \g^{(k)}=\g^{(k)}_{(2)} +\cdots +
\g^{(k)}_{(k+1)},
\quad k=1,\ldots, N+1,\qquad \g^{(N+2)}=0, \\
&& \g^{r_{k-1}}_{(i)}= \op\sum_{k_1+\cdots+ k_i=k+1
-i}(\op\sum_{0\leq \La_{k_1},\ldots, \La_{k_i}}
\g^{r_{k-1},\La_{k_1},\ldots,\La_{k_i}}_{(i)r_{k_1},\ldots,r_{k_i}}c^{r_{k_1}}_{\La_{k_1}}
\cdots c^{r_{k_i}}_{\La_{k_i}}), \quad k=1,\ldots, N+1,
\ee
where $\g^{(k)}_{(i)}$ are terms of polynomial degree $2\leq i\leq
k+1$ in ghosts. Then the nilpotent property (\ref{851}) of $\bb$
falls into a set of equalities
\mar{w110,3}\ben
&& u^{(k+1)}(u^{(k)})
=0, \qquad 0\leq k\leq N-1,  \label{w110}\\
&& (u +\g^{(k+1)}_{(2)})(u^{(k)}) + u_{HS}(\g^{(k)}_{(2)})=0,
\qquad 0\leq k\leq N+1, \label{w111} \\
&& \g_{(i)}^{(k+1)}(u^{(k)}) + u (\g_{(i-1)}^{(k)}) +
u_{HS}(\g^{k}_{(i)}) + \label{w113}\\
&&\qquad \op\sum_{2\leq m\leq
i-1}\g_{(m)}(\g_{(i-m+1)}^{(k)})=0, \qquad  i-2\leq k\leq N+1,
\nonumber
\een
of ghost polynomial degree 1, 2 and $3\leq i\leq N+3$,
respectively.

The equalities (\ref{w110}) are exactly the gauge symmetry
conditions (\ref{850}) in Proposition \ref{lmp6}.

The equality (\ref{w111}) for $k=0$ reads
\mar{852}\beq
(u+ \g^{(1)})(u)=0, \qquad \op\sum_{0\leq |\La|} (d_\La(u^A)
\dr_A^\La u^B+ d_\La(\g^r)u^{B,\La}_r)=0. \label{852}
\eeq
It takes the form of the Lie antibracket
\mar{s12}\beq
[u,u]=-2\g^{(1)}(u)=-2\op\sum_{0\leq |\La|}
d_\La(\g^r)u^{B,\La}_r\dr_B \label{s12}
\eeq
of the odd gauge symmetry $u$. Its right-hand side is a non-linear
differential operator on the module $\cC^{(0)}$ of ghosts taking
values in the real space $\cG_L$ of variational symmetries.
Following Definition \ref{s8}  extended to Lagrangian theories of
odd fields, we treat it as a generalized gauge symmetry
factorizing through the gauge symmetry $u$. Thus, we come to the
following.

\begin{prop} \label{830} \mar{830}
The gauge operator (\ref{w108}) admits the nilpotent extension
(\ref{w109}) only if the Lie antibracket of the odd gauge symmetry
$u$ (\ref{w33}) is a generalized gauge symmetry factorizing
through $u$.
\end{prop}

The equalities (\ref{w111}) -- (\ref{w113}) for $k=1$ take the
form
\mar{853,4}\ben
&& (u +\g^{(2)}_{(2)})(u^{(1)}) + u^{(1)}(\g^{(1)})=0, \label{853} \\
&& \g_{(3)}^{(2)}(u^{(1)}) + (u + \g^{(1)})(\g^{(1)})=0. \label{854}
\een
In particular, if Lagrangian field theory is irreducible, i.e.,
$u^{(k)}=0$ and $\bu=u$, the BRST operator reads
\be
 \bb= u+\g^{(1)}=u^A\dr_A + \g^r\dr_r=
\op\sum_{0\leq|\La|} u^{A,\La}_r c^r_\La \dr_A +
\op\sum_{0\leq|\La|,|\Xi|}\g^{r,\La,\Xi}_{pq}c^p_\La c^q_\Xi
\dr_r,
\ee
and the nilpotency conditions (\ref{853}) - (\ref{854}) are
reduced to the equality
\mar{0691}\beq
(u + \g^{(1)})(\g^{(1)})=0. \label{0691}
\eeq
Furthermore, let a gauge symmetry $u$ be affine in fields $s^A$
and their jets. Then it follows from the nilpotency condition
(\ref{852}) that the BRST term $\g^{(1)}$ is independent of
original fields and their jets.  Then the relation (\ref{0691})
takes the form of the Jacobi identity
\mar{s11}\beq
\g^{(1)})(\g^{(1)})=0 \label{s11}
\eeq
for coefficient functions $\g^{r,\La,\Xi}_{pq}(x)$ in the Lie
antibracket (\ref{s12}).

The relations (\ref{s11}) and (\ref{s12}) motivate us to think of
the equalities (\ref{w111}) -- (\ref{w113}) in a general case of
reducible gauge symmetries as being {\it sui generis} commutation
relations and Jacobi identities, respectively. Moreover, based on
Proposition \ref{830}, we say that non-trivial gauge symmetries
are algebraically closed (in the terminology of Ref. [5])  if the
gauge operator $\bu$ (\ref{w108'}) admits the nilpotent BRST
extension $\bb$ (\ref{w109}).

\bigskip
\bigskip

\noindent {\bf VI. YANG--MILLS SUPERGAUGE THEORY}
\bigskip

Yang--Mills supergauge theory exemplifies theory of odd fields.

Let $\ccG=\ccG_0\oplus \ccG_1$ be a finite-dimensional real Lie
superalgebra with the basis $\{e_r\}$, $r=1,\ldots,m,$ and real
structure constants $c^r_{ij}$. Recall that
\be
&& c^r_{ij}=-(-1)^{[i][j]}c^r_{ji}, \qquad [r]=[i]+[j],\\
&& (-1)^{[i][b]}c^r_{ij}c^j_{ab} + (-1)^{[a][i]}c^r_{aj}c^j_{bi} +
(-1)^{[b][a]}c^r_{bj}c^j_{ia}=0,
\ee
where $[r]$ denotes the Grassmann parity of $e_r$. Given the
universal enveloping algebra $\ol \ccG$ of $\ccG$, we assume that
there is an even quadratic Casimir element $h^{ij}e_ie_j$ of
$\ol\ccG$ such that the matrix $h^{ij}$ is non-degenerate. The
Yang--Mills theory on $X=\Bbb R^n$ associated to this Lie
superalgebra is described by the DGA $\cP^*_\infty[F;Y]$ where
\be
F=\ccG\op\ot_X T^*X, \qquad Y= \ccG_0\op\ot_X T^*X.
\ee
Its local basis is $(a^r_\la)$, $[a^r_\la]=[r]$. First jets of its
elements admit the canonical splitting
\mar{f31}\beq
a^r_{\la\m}=\frac12(\cF^r_{\la\m} +
\cS^r_{\la\m})=\frac12(a^r_{\la\m}-a^r_{\m\la} +c^r_{ij}a^i_\la
a^j_\m) +\frac12(a^r_{\la\m}+ a^r_{\m\la} -c^r_{ij}a^i_\la
a^j_\m). \label{f31}
\eeq
Given a constant metric $g$ on $\Bbb R^n$, the Yang--Mills
Lagrangian reads
\be
L_{YM}=\frac14
h_{ij}g^{\la\m}g^{\bt\nu}\cF^i_{\la\bt}\cF^j_{\m\nu}d^nx.
\ee
Its variational derivatives $\cE_r^\la$ obey the irreducible NI
\be
 c^r_{ji}a^i_\la\cE_r^\la + d_\la\cE_j^\la=0.
\ee
Therefore, we enlarge the DGA $\cP^*_\infty[F;Y]$ to the DGA
\be
P^*_\infty\{0\}=\cP^*_\infty[F\op\oplus_Y E_0;Y], \qquad
E_0=X\times (\ccG_1\oplus \ccG_0),
\ee
whose local basis $(a^r_\la, c^r)$, $[c^r]= ([r]+1){\rm mod}\,2$,
contains ghosts $c^r$ of ghost number 1. Then the gauge operator
(\ref{w108'}) reads
\be
\bu= (-c^r_{ji}c^ja^i_\la + c^r_\la)\frac{\dr}{\dr a_\la^r}.
\ee
It admits the nilpotent BRST extension
\be
\bb=\bu +\xi= (-c^r_{ji}c^ja^i_\la + c^r_\la)\frac{\dr}{\dr
a_\la^r} -\frac12 (-1)^{[i]} c^r_{ij}c^ic^j\frac{\dr}{\dr c^r}.
\ee

\bigskip
\bigskip

\noindent {\bf VII. TOPOLOGICAL CHERN--SIMONS THEORY}
\bigskip

We consider gauge theory of principal connections on a principal
bundle $P\to X$ with a structure real Lie group $G$. In contrast
with the Yang--Mills Lagrangian, the Chern--Simons (henceforth CS)
Lagrangian is independent of a metric on $X$. Therefore, its
non-trivial gauge symmetries are wider than those of the
Yang--Mills one. Moreover, some of them become trivial if $\di
X=3$.

Note that one usually considers the local CS Lagrangian which is
the local CS form derived from the local transgression formula for
the Chern characteristic form. The global CS Lagrangian is well
defined, but depends on a background gauge potential.$^{16-18}$

The fiber bundle $J^1P\to C$ is a trivial $G$-principal bundle
canonically isomorphic to $C\times P\to C$.$^1$ This bundle admits
the canonical principal connection
\be
\cA =dx^\la\ot(\dr_\la +a_\la^p \ve_p) + da^r_\la\ot\dr^\la_r.
\ee
Its curvature defines the canonical $(VP/G)$-valued 2-form
\mar{f34}\beq
\gF =(da_\m^r\w dx^\m + \frac{1}{2} c_{pq}^r a_\la^p a_\m^q
dx^\la\w dx^\m)\ot e_r \label{f34}
\eeq
on $C$. Given a section $A$ of $C\to X$, the pull-back
\be
F_A=A^*\gF=\frac12 F^r_{\la\m}dx^\la\w dx^\m\ot e_r
\ee
of $\gF$ onto $X$ is the strength form of a gauge potential $A$.

Let $I_k(\chi)=b_{r_1\ldots r_k}\chi^{r_1}\cdots \chi^{r_k}$ be a
$G$-invariant polynomial of degree $k>1$ on the Lie algebra $\ccG$
of $G$. With $\gF$ (\ref{f34}), one can associate to $I_k$ the
closed $2k$-form
\mar{0757}\beq
P_{2k}(\gF)=b_{r_1\ldots r_k}\gF^{r_1}\w\cdots\w \gF^{r_k}, \qquad
k\leq \di X, \label{0757}
\eeq
on $C$ which is invariant under automorphisms of $C$ induced by
vertical automorphisms of $P$. Given a section $B$ of $C\to X$,
the pull-back $ P_{2k}(F_B)=B^*P_{2k}(\gF)$ of $P_{2k}(\gF)$ is a
closed characteristic form on $X$. Let the same symbol stand for
its pull-back onto $C$. Since $C\to X$ is an affine bundle and the
de Rham cohomology of $C$ equals that of $X$, the forms
$P_{2k}(\gF)$ and $P_{2k}(F_B)$ possess the same cohomology class
$[P_{2k}(\gF)]=[P_{2k}(F_B)]$ for any principal connection $B$.
Thus, $I_k(\chi)\mapsto [P_{2k}(F_B)]\in H^*_{DR}(X)$ is the
familiar Weil homomorphism. Furthermore, we obtain the
transgression formula
\mar{r65}\beq
P_{2k}(\gF)-P_{2k}(F_B)=d\gS_{2k-1}(a,B)\label{r65}
\eeq
on $C$.$^{16}$ Its pull-back by means of a section $A$ of $C\to X$
gives the transgression formula
\be
P_{2k}(F_A)-P_{2k}(F_B)=d \gS_{2k-1}(A,B)
\ee
on $X$. For instance, if $P_{2k}(\gF)$ is the characteristic Chern
$2k$-form, then $\gS_{2k-1}(a,B)$ is the CS $(2k-1)$-form. In
particular, one can choose the local section $B=0$. Then,
$\gS_{2k-1}(a,0)$ is the local CS form. Let $\gS_{2k-1}(A,0)$ be
its pull-back onto $X$ by means of a section $A$ of $C\to X$. Then
the CS form $\gS_{2k-1}(a,B)$ (\ref{r65}) admits the decomposition
\mar{r75}\beq
\gS_{2k-1}(a,B)=\gS_{2k-1}(a,0) -\gS_{2k-1}(B,0) +dK_{2k-1}.
\label{r75}
\eeq
The transgression formula (\ref{r65}) also yields the
transgression formula
\mar{0742}\ben
&& P_{2k}(\cF)-P_{2k}(F_B)=d_H(h_0 \gS_{2k-1}(a,B)), \nonumber\\
&& h_0 \gS_{2k-1}(a,B)=k\op\int^1_0 \cP_{2k}(t,B)dt, \label{0742}\\
&& \cP_{2k}(t,B)=b_{r_1\ldots
r_k}(a^{r_1}_{\m_1}-B^{r_1}_{\m_1})dx^{\m_1}\w
\cF^{r_2}(t,B)\w\cdots \w \cF^{r_k}(t,B),\nonumber\\
&& \cF^{r_j}(t,B)= \frac12[ ta^{r_j}_{\la_j\m_j}
+(1-t)\dr_{\la_j}B^{r_j}_{\m_j}
- ta^{r_j}_{\m_j\la_j} -(1-t)\dr_{\m_j}B^{r_j}_{\la_j}+\nonumber\\
&& \qquad \frac12c^{r_j}_{pq} (ta^p_{\la_j}
+(1-t)B^p_{\la_j})(ta^q_{\m_j} +(1-t)B^q_{\m_j}]dx^{\la_j}\w
dx^{\m_j}\ot e_r,\nonumber
\een
on $J^1C$. If $2k-1=\di X$, the density
$L_{CS}(B)=h_0\gS_{2k-1}(a,B)$ (\ref{0742}) is the global CS
Lagrangian of topological CS theory. The decomposition (\ref{r75})
induces the decomposition
\mar{0747}\beq
L_{CS}(B)=h_0\gS_{2k-1}(a,0) -\gS_{2k-1}(B,0) +d_H h_0K_{2k-1}.
\label{0747}
\eeq

For instance, if $\di X=3$, the global CS Lagrangians reads
\mar{s20}\ben
&& L_{CS}(B)= [\frac12 h_{mn} \ve^{\al\bt\g}a^m_\al(\cF^n_{\bt\g}
-\frac13 c^n_{pq}a^p_\bt a^q_\g)]d^nx - \label{s20} \\
&& \qquad [\frac12 h_{mn} \ve^{\al\bt\g}B^m_\al(F(B)^n_{\bt\g}
-\frac13 c^n_{pq}B^p_\bt B^q_\g)]d^nx -d_\al(h_{mn}
\ve^{\al\bt\g}a^m_\bt B^n_\g)d^nx, \nonumber
\een
where $\ve^{\al\bt\g}$ is the skew-symmetric Levi--Civita tensor.

Since the density $-\gS_{2k-1}(B,0) +d_Hh_0K_{2k-1}$ is
variationally trivial, the global CS Lagrangian (\ref{0747})
possesses the same NI and gauge symmetries as the local one
$L_{CS}=h_0\gS_{2k-1}(a,0)$. They are the following.

Infinitesimal generators of local one-parameter groups of
automorphisms of a principal bundle $P$ are $G$-invariant
projectable vector fields $v_P$ on $P$. They are identified with
sections
\mar{0745}\beq
v_P=\tau^\la\dr_\la +\xi^r e_r \label{0745}
\eeq
of the vector bundle $T_GP=TP/G\to X$, and yield vector fields
\mar{0653}\beq
v_C=\tau^\la\dr_\la +(-c^r_{pq}\xi^pa^q_\la +\dr_\la \xi^r
-a^r_\m\dr_\la \tau^\m)\dr^\la_r \label{0653}
\eeq
on the bundle of principal connections $C$.$^1$  Sections $v_P$
(\ref{0745}) play a role of gauge parameters. One can show that
vector fields (\ref{0653}) are variational symmetries of the
global CS Lagrangian $L_{CS}(B)$. By virtue of Proposition
\ref{812}, the vertical part
\mar{0785}\beq
v_V=(-c^r_{pq}\xi^pa^q_\la +\dr_\la \xi^r -a^r_\m\dr_\la \tau^\m
-\tau^\m a^r_{\m\la} )\dr^\la_r \label{0785}
\eeq
of a vector field $v_C$ (\ref{0653}) is also a variational
symmetry of $L_{CS}(B)$.

Let us consider the DGA $\cP^*_\infty[T_GP;C]$ possessing the
local basis $(a^r_\la, c^\la, c^r)$ of even fields $a^r_\la$ and
odd ghosts $c^\la$, $c^r$. Substituting these ghosts for gauge
parameters in the vector field $v_V$ (\ref{0785}), we obtain the
odd vertical graded derivation
\mar{0781}\beq
u=(-c^r_{pq}c^pa^q_\la + c^r_\la -c^\m_\la a^r_\m -c^\m
a_{\m\la}^r)\dr^\la_r \label{0781}
\eeq
of the DGA $\cP^*_\infty[T_GP;C]$. This graded derivation as like
as vector fields $v_V$ (\ref{0785}) is a variational and,
consequently, gauge symmetry of the CS Lagrangian $L_{CS}(B)$. By
virtue of the formulas (\ref{0656}) -- (\ref{0657}), the
corresponding NI read
\mar{s15}\beq
\ol\dl\Delta_j= -c^r_{ji}a^i_\la\cE_r^\la - d_\la\cE_j^\la=0,
\qquad \ol\dl\Delta_\m=-
 a^r_{\m\la}\cE^\la_r +d_\la(a^r_\m\cE^\la_r)=0. \label{s15}
\eeq
They are irreducible non-trivial, unless $\di X=3$. Therefore, the
gauge operator (\ref{w108}) is $\bu=u$. It admits the nilpotent
BRST extension
\be
\bb=(-c^r_{ji}c^ja^i_\la + c^r_\la -c^\m_\la a^r_\m -c^\m
a_{\m\la}^r)\frac{\dr}{\dr a_\la^r} -\frac12
c^r_{ij}c^ic^j\frac{\dr}{\dr c^r} +c^\la_\m c^\m\frac{\dr}{\dr
c^\la}.
\ee

If $\di X=3$, the CS Lagrangian takes the form (\ref{s20}), and
the corresponding Euler--Lagrange operator reads
\be
\dl L_{CS}(B)=\cE^\la_r \th^r_\la\w d^nx, \qquad \cE^\la_r=h_{rp}
\ve^{\la\bt\g}\cF^p_{\bt\g}.
\ee
A glance at the NI (\ref{s15}) shows that they are equivalent to
NI
\mar{s16}\beq
\ol\dl\Delta_j=-c^r_{ji}a^i_\la\cE_r^\la - d_\la\cE_j^\la=0,\qquad
\ol\dl\Delta'_\m= \ol\dl\Delta_\m
+a^r_\m\ol\dl\Delta_r=c^\m\cF^r_{\la\m}\cE^\la_r=0. \label{s16}
\eeq
These NI define the gauge symmetry $u$ (\ref{0781}) written in the
form
\mar{s18}\beq
u=(-c^r_{pq}c'^pa^q_\la + c'^r_\la +c^\m\cF^r_{\la\m})\dr^\la_r
\label{s18}
\eeq
where $c'^r=c^r-a^r_\m c^\m$. It is readily observed that, if $\di
X=3$, the NI $\ol\dl\Delta'_\m$ (\ref{s16}) are trivial. Then the
corresponding part $c^\m\cF^r_{\la\m})\dr^\la_r$ of the gauge
symmetry $u$ (\ref{s18}) is also trivial. Consequently, the
non-trivial gauge symmetry of the CS Lagrangian (\ref{s20}) is
\be
u=(-c^r_{pq}c'^pa^q_\la + c'^r_\la)\dr^\la_r.
\ee

\bigskip
\bigskip

\noindent {\bf VIII. GAUGE GRAVITATION THEORY}
\bigskip

Gravitation theory can be formulated as gauge theory on natural
bundles $T$ over an oriented four-dimensional manifold
$X$.$^{19,20}$ It is  metric-affine gravitation theory whose
Lagrangian $L_{MA}$ is invariant under general covariant
transformations. Infinitesimal generators of local one-parameter
groups of these transformations are the functorial lift (i.e., the
Lie algebra monomorphism) of vector fields on $X$ onto a natural
bundle. Vector fields on $X$ are gauge parameters of general
covariant transformations. Natural bundles are exemplified by
tensor bundles over $X$. The principal bundle $LX$ of linear
frames in the tangent bundle $TX$ of $X$ and associated tensor
bundles exemplify natural bundle.

Dynamic variables of gauge gravitation theory on natural bundles
are linear connections and pseudo-Riemannian metrics on $X$.
Linear connections on $X$ are principal connections on the
principal frame bundle $LX$ with the structure group
$GL_4=GL^+(4,\Bbb R)$. They are represented by global sections of
the quotient bundle $C_K=J^1LX/GL_4$. It is also a natural bundle
provided with bundle coordinates $(x^\la,k_\la{}^\nu{}_\al)$ such
that components $k_\la{}^\nu{}_\al\circ K=K_\la{}^\nu{}_\al$ of a
section $K$ of $C_K\to X$ are coefficient of the linear connection
\be
K=dx^\la\ot (\dr_\la + K_\la{}^\m{}_\nu \dot x^\nu\dot\dr_\mu)
\ee
on $TX$ with respect to the holonomic bundle coordinates
$(x^\la,\dot x^\la)$. The first order jet manifold $J^1C_K$ of
$C_K$ admits the canonical decomposition taking the coordinate
form
\be
&& k_{\la\m}{}^\al{}_\bt=\frac12(R_{\la\m}{}^\al{}_\bt +
S_{\la\m}{}^\al{}_\bt)=\frac12(k_{\la\m}{}^\al{}_\bt -
k_{\m\la}{}^\al{}_\bt + k_\m{}^\al{}_\ve k_\la{}^\ve{}_\bt
-k_\la{}^\al{}_\ve k_\m{}^\ve{}_\bt)+  \\
&& \qquad \frac12(k_{\la\m}{}^\al{}_\bt + k_{\m\la}{}^\al{}_\bt -
k_\m{}^\al{}_\ve k_\la{}^\ve{}_\bt +k_\la{}^\al{}_\ve
k_\m{}^\ve{}_\bt).
\ee
If $K$ is a section of $C_K\to X$, then $R\circ K$ is the
curvature of a linear connection $K$.

In gravitation theory, the linear frame bundle $LX$ is assumed to
admit a Lorentz structure, i.e., reduced principal subbundles with
the structure Lorentz group $SO(1,3)$. By virtue of the well-known
theorem, there is one-to-one correspondence between these
subbundles and the global sections of the quotient bundle
$\Si=LX/SO(1,3)$. Its sections are pseudo-Riemannian metrics on
$X$. Being an open subbundle of the tensor bundle $\op\vee^2 TX$,
the bundle $\Si$ is provided with bundle coordinates
$(x^\la,\si^{\m\nu})$.

The total configuration space of gauge gravitation theory is the
bundle product $\Si\times C_K$ coordinated by
$(x^\la,\si^{\al\bt},  k_\mu{}^\al{}_\bt)$. This is a natural
bundle admitting the functorial lift
\mar{gr3}\ben
&& \wt\tau_{\Si K}=\tau^\m\dr_\m +(\si^{\nu\bt}\dr_\nu \tau^\al
+\si^{\al\nu}\dr_\nu \tau^\bt)\frac{\dr}{\dr \si^{\al\bt}} +
\label{gr3}\\
&& \qquad (\dr_\nu \tau^\al k_\m{}^\nu{}_\bt -\dr_\bt \tau^\nu
k_\m{}^\al{}_\nu -\dr_\mu \tau^\nu k_\nu{}^\al{}_\bt
+\dr_{\m\bt}\tau^\al)\frac{\dr}{\dr k_\mu{}^\al{}_\bt} \nonumber
\een
of vector fields $\tau=\tau^\m\dr_\m$ on $X$.$^{1,13}$ Let us
consider the DGA $\cS^*_\infty[\Si\times C_K]$ possessing the
local basis $(\si^{\al\bt},  k_\mu{}^\al{}_\bt)$, and let us
enlarge it to the DGA
\mar{07100}\beq
\cP^*_\infty[TX; \Si\times C_K] \label{07100}
\eeq
possessing the local basis $(\si^{\al\bt},  k_\mu{}^\al{}_\bt,
c^\m)$ of even fields $(\si^{\al\bt}, k_\mu{}^\al{}_\bt)$ and odd
ghosts $(c^\m)$. Taking the vertical part of vector fields
$\wt\tau_{K\Si}$ (\ref{gr3}) and replacing gauge parameters
$\tau^\la$ with ghosts $c^\la$, we obtain the odd vertical graded
derivation
\be
&&u=u^{\al\bt}\frac{\dr}{\dr\si^{\al\bt}} +u_\m{}^\al{}_\bt
\frac{\dr}{\dr k_\mu{}^\al{}_\bt} =(\si^{\nu\bt} c_\nu^\al
+\si^{\al\nu} c_\nu^\bt-c^\la\si_\la^{\al\bt})\frac{\dr}{\dr
\si^{\al\bt}}+
\\
&& \qquad (c_\nu^\al k_\m{}^\nu{}_\bt -c_\bt^\nu k_\m{}^\al{}_\nu
-c_\mu^\nu k_\nu{}^\al{}_\bt +c_{\m\bt}^\al-c^\la
k_{\la\mu}{}^\al{}_\bt)\frac{\dr}{\dr k_\mu{}^\al{}_\bt}
\ee
of the DGA (\ref{07100}). It is a gauge symmetry of a gravitation
Lagrangian $L_{MA}$. Then by virtue of the formulas (\ref{0656})
-- (\ref{0657}), the Euler--Lagrange operator
\be
(\cE_{\al\bt} d\si^{\al\bt} + \cE^\m{}_\al{}^\bt
dk_\m{}^\al{}_\bt)\w d^4x
\ee
of this Lagrangian obeys the NI
\be
&& -\si^{\al\bt}_\la \cE_{\al\bt} - 2d_\m(\si^{\m\bt}\cE_{\la\bt} -
k_{\la\m}{}^\al{}_\bt\cE^\m{}_\al{}^\bt- \\
&&\qquad d_\m[(k_\nu{}^\m{}_\bt\dl^\al_\la - k_\nu{}^\al{}_\la \dl^\m_\bt
- k_\la{}^\al{}_\bt \dl^\m_\nu)\cE^\nu{}_\al{}^\bt] +
d_{\m\bt}\cE^\m{}_\la{}^\bt=0.
\ee
These NI are irreducible. Therefore, the gauge operator
(\ref{w108}) is $\bu=u$. Its BRST extension reads$^{13}$
\be
\bb=u + c^\la_\m c^\m\frac{\dr}{\dr c^\la}.
\ee
Note that this BRST operator differs from that in Ref. [21], where
metric-affine gravitation theory is treated as gauge theory of the
Poincar\'e group.$^{22-24}$

\bigskip
\bigskip

\noindent {\bf IX. TOPOLOGICAL BF THEORY}
\bigskip

We address the topological BF theory of two exterior forms $A$ and
$B$ of form degree $|A|+|B|=\di X-1$ on a smooth manifold
$X$.$^{25}$ It is reducible degenerate Lagrangian theory which
satisfies the homology regularity condition.$^{3,4}$ Its dynamic
variables are exterior forms $A$ and $B$ of form degree $|A|
+|B|=n-1$ on a manifold $X$. They are sections of the bundle
\be
Y=\op\w^pT^*X\oplus \op\w^qT^*X, \qquad p+q=n-1,
\ee
coordinated by $(x^\la, A_{\m_1\ldots\m_p},B_{\nu_1\ldots\nu_q})$.
Without a loss of generality, let $q$ be even and $q\geq p$. The
corresponding DGA is $\cO^*_\infty Y$. There are the canonical
$p$- and $q$-forms
\be
A=\frac{1}{p!}A_{\m_1\ldots\m_p}dx^{\m_1}\w\cdots\w dx^{\m_p},
\qquad
B=\frac{1}{q!}B_{\nu_{p+1}\ldots\nu_q}dx^{\nu_{p+1}}\w\cdots\w
dx^{\nu_p}
\ee
on $Y$. A Lagrangian of the topological BF theory reads
\be
L_{\rm BF}=A\w d_HB.
\ee
Its Euler--Lagrange equations $d_HA=0$, $d_HB=0$ obey the NI
\be
d_Hd_HA= 0, \qquad d_Hd_H B= 0.
\ee

Given the vector bundles
\be
&& E_k=\op\w^{p-k-1}T^*X\op\times_X \op\w^{q-k-1}T^*X, \qquad 0\leq
k< p-1, \\
&& E_k=\Bbb R \op\times_X
\op\w^{q-p}T^*X, \qquad k=p-1, \\
&& E_k=\op\w^{q-k-1}T^*X, \quad p-1<k<q-1, \\
&& E_{q-1}=X\times \Bbb R,
\ee
let us consider the DGA $P_\infty^*\{q-1\}$ with the local basis
\be
&& \{A_{\m_1\ldots\m_p}, B_{\nu_{p+1}\ldots\nu_q},
\ve_{\m_2\ldots\m_p},\ldots,\ve_{\m_p},\ve,\xi_{\nu_{p+2}\ldots\nu_q},
\ldots, \xi_{\nu_q},\xi,\\
&&\qquad \ol A^{\m_1\ldots\m_p}, \ol B^{\nu_{p+1}\ldots\nu_q},
\ol\ve^{\m_2\ldots\m_p}, \ldots,\ol\ve^{\m_p}, \ol \ve, \ol
\xi^{\nu_{p+2}\ldots\nu_q}, \ldots, \ol \xi^{\nu_q},\ol \xi\}.
\ee
Then the gauge operator (\ref{w108'}) reads
\be
&& \bu= d_{\m_1}\ve_{\m_2\ldots\m_p}\frac{\dr}{\dr
A_{\m_1\m_2\ldots\m_p}} +
d_{\nu_{p+1}}\xi_{\nu_{p+2}\ldots\nu_q}\frac{\dr}{\dr
B_{\nu_{p+1}\nu_{p+2}\ldots\nu_q}}+
[d_{\m_2}\ve_{\m_3\ldots\m_p}\frac{\dr}{\dr
\ve_{\m_2\m_3\ldots\m_p}}+\cdots  \\
&& \qquad  +d_{\m_p}\ve\frac{\dr}{\dr \ve^{\m_p}}]+
[d_{\nu_{p+2}}\xi_{\nu_{p+3}\ldots\nu_q} \frac{\dr}{\dr
\xi_{\nu_{p+2}\nu_{p+3}\ldots\nu_q}}+\cdots +
d_{\nu_q}\xi\frac{\dr}{\dr \xi^{\nu_q}}].
\ee
This operator is obviously nilpotent and, thus, is the BRST
operator $\bb=\bu$.

\end{document}